\documentclass[conference]{IEEEtran}
\IEEEoverridecommandlockouts
\usepackage{cite}
\usepackage{amsmath,amssymb,amsfonts,amsthm}
\usepackage{algorithm}
\usepackage[noend]{algpseudocode}
\usepackage{graphicx}
\usepackage{textcomp}
\usepackage{xcolor}
\usepackage{subfigure}
\usepackage{cite}
\usepackage{tikz-network}
\usepackage{todonotes}

\def\BibTeX{{\rm B\kern-.05em{\sc i\kern-.025em b}\kern-.08em
    T\kern-.1667em\lower.7ex\hbox{E}\kern-.125emX}}
    
\newtheorem{defn}{Definition}
\newtheorem{thm}{Theorem}

\newtheorem{prop}{Proposition}
\newtheorem{cor}{Corollary}

\newtheorem{rem}{Remark}
\newtheorem{exam}{Example}
\newtheorem{assump}{Assumption}
\newtheorem{prob}{Problem}

\newcommand{\G}{\mathcal{G}}
\newcommand{\Aut}{\mathrm{Aut}}

\newcommand{\EE}{\mathbb{E}}
\newcommand{\V}{\mathbb{V}}

\newcommand{\E}{\mathcal{E}}
\newcommand{\GG}{\mathbb{G}}

\newcommand{\lcm}{\mathrm{lcm}}

\begin{document}

\title{Cluster Assignment in Multi-Agent Systems\\

\thanks{This research was supported by grant no. 2285/20 from the Israel Science Foundation}
}

\author{\IEEEauthorblockN{%1\textsuperscript{st} 
Miel Sharf}
\IEEEauthorblockA{\textit{School of Electrical Engineering and Computer Science} \\
\textit{KTH Royal Institute of Technology}\\
Stockholm, Sweden \\
sharf@kth.se}
\and
\IEEEauthorblockN{%2\textsuperscript{nd} 
Daniel Zelazo}
\IEEEauthorblockA{\textit{Faculty of Aerospace Engineering} \\
\textit{Technion - Israel Institute of Technology}\\
Haifa, Israel \\
dzelazo@technion.ac.il}

}
\maketitle

\begin{abstract}
We study cluster assignment in multi-agent networks. We consider homogeneous diffusive networks, and focus on design of the graph that ensures the system will converge to a prescribed cluster configuration, i.e., specifying the number of clusters and agents within each cluster. Leveraging recent results from cluster synthesis, we show that it is possible to design an oriented graph such that the action of the automorphism group of the graph has orbits of predetermined sizes, guaranteeing that the network will converge to the prescribed cluster configuration. We provide upper and lower bounds on the number of edges that are needed to construct these graphs along with a constructive approach for generating these graphs. We support our analysis with some numerical examples.
\end{abstract}


\section{Introduction}\label{sec.intro}

The process of reaching an agreement between agents is one of the fundamental tasks for a multi-agent system (MAS).  Indeed, agreement protocols, the decision rule implemented by each agent that enables them to distributedly reach agreement, appear across many diverse fields.  These include distributed computation \cite{Xiao2004}, robotics \cite{Chopra2006}, biochemical systems \cite{Scardovi2010}, and sensor networks \cite{OlfatiSaber2007cdc}.  A natural extension to the agreement problem is the \emph{cluster agreement} problem, which seeks to drive agents into groups.  All the agents within the same group should then reach an agreement.  The clustering problem also appears in many areas including neuroscience \cite{Schnitzler2005}, biomimicry of swarms \cite{Passino2002}, and social networks \cite{Lancichinetti2012}.

Various approaches have been used to study clustering, e.g. exploiting the structural balance of the underlying graph \cite{Altafini2013}, pinning control \cite{Qin2013}, inter-cluster nonidentical inputs \cite{Han2013}, and network optimizaton \cite{Burger2011TAC}.  Our approach to the cluster agreement problem is to leverage notions of \emph{symmetry} within a multi-agent system.  Symmetry of graphs has recently emerged as an important property for multi-agent systems, in particular in the study of controllability and observability properties of these systems \cite{Rahmani2007, Chapman2014,Chapman2015, Notarstefano2013TAC}. 

In our recent work \cite{Sharf2019b}, we introduced the notion of the \emph{weak automorphism group of an MAS}.  This new notion of symmetry for MAS combines two ideas: the automorphism group for graphs and weak equivalence of dynamical systems.  Graph automorphisms are known to capture notions of symmetries for graphs, while {weak equivalence} of dynamical systems aims to characterize similarities of heterogeneous dynamical agents in terms of their achievable steady-states. The weak automorphisms of an MAS can be thought of, therefore, as a permutation of the nodes in the underlying graph that preserves graph symmetries together with certain input-output  properties (i.e., steady-state maps) of the associated agents.  We focused on clustering for \emph{diffusively coupled networks}, as shown in Figure \ref{ClosedLoop}.  Under appropriate passivity assumptions for the agent dynamics, we showed that these diffusively coupled networks converge to a clustered solution in the steady-state, where two agents are in the same cluster if and only if there exists a weak automorphism mapping one to the other. Thus, the clustering of the MAS can be understood by studying the action of the weak automorphism group on the underlying interaction graph.

In this work we focus on homogeneous networks, that is networks where the agent dynamics are all identical, noting that the weak automorphism group of the network is identical to the automorphism group of the underlying graph in this case.  The problem we aim to solve is how to design graphs that ensure the networked system will converge to a prescribed cluster configuration, i.e., specifying the number of clusters and the number of agents within each cluster.  Employing tools from group theory, we show that it is possible to design an oriented graph such that the action of the automorphism group of the graph has orbits of specified sizes determined by the designer.  We provide upper and lower bounds on the number of edges that are needed to construct these graphs.  The proof of these results in turn provide a constructive approach for generating these graphs.  We support our analysis with some numerical examples.  

The rest of the paper is organized as  follows.  Section \ref{sec.background} reviews basic concepts related to network systems and group theory required to define a notion of symmetry for multi-agent systems.  In Section \ref{sec.cluster}, the main results related to cluster assignment are given.  This section also presents some numerical studies to demonstrate the theory. Finally, some concluding remarks are offered in Section \ref{sec.conclusion}.

\paragraph*{Notations}
This work employs basic notions from graph theory \cite{Godsil2001}.  An undirected graph $\mathcal{G}=(\mathbb{V},\mathbb{E})$ consists of a finite set of vertices $\mathbb{V}$ and edges $\mathbb{E} \subset \mathbb{V} \times \mathbb{V}$.  We denote by $k=\{i,j\} \in \mathbb{E}$ the edge that has ends $i$ and $j$ in $\mathbb{V}$. For each edge $k$, we pick an arbitrary orientation and denote $k=(i,j)$ when $i\in \mathbb{V}$ is the \emph{head} of edge $k$ and $j \in \mathbb{V}$ the \emph{tail}. A path is a sequence of distinct nodes $v_1, v_2,\ldots , v_n$ such that $\{v_i, v_{i+1}\} \in \EE$ for all i. A cycle is the union of a path $v_1,\ldots,v_n$ with the edge $\{v_1,v_n\}$. A simple cycle is a cycle whose vertices are all distinct. A graph is called connected if there is a path between any two vertices, and a graph is called a tree if it is connected but contains no simple cycles.

The incidence matrix of $\mathcal{G}$, denoted $\mathcal{E}\in\mathbb{R}^{|\mathbb{E}|\times|\mathbb{V}|}$, is defined such that for edge $k=(i,j)\in \mathbb{E}$, $[\mathcal{E}]_{ik} =+1$, $[\mathcal{E}]_{jk} =-1$, and $[\mathcal{E}]_{\ell k} =0$ for $\ell \neq i,j$. For a graph $\G$, an automorphism of $\G$ is a permutation $\psi:\V\to\V$ such that $i$ is connected to $j$ if and only if $\psi(i)$ is connected to $\psi(j)$. We denote its automorphism group by $\Aut(\G)$. Lastly, the greatest common divisor of two positive integers $m,n$ is denoted by $\gcd(m,n)$, and their least common multiple is denoted by $\lcm(m,n)$. Note that $\lcm(m,n) = \frac{mn}{\gcd(m,n)}$ always holds.
\begin{figure} [!t] 
    \centering
    \includegraphics[scale=0.7]{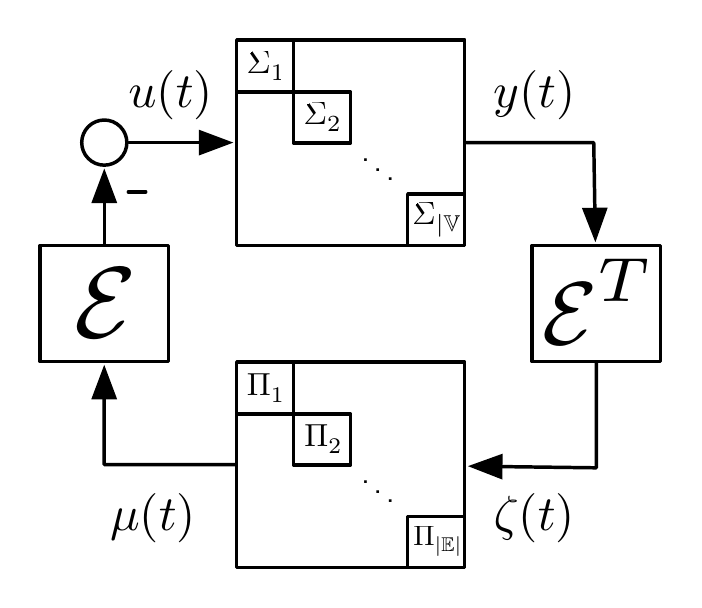}
    \vspace{-10pt}
    \caption{A diffusively coupled network.}
    \vspace{-10pt}
    \label{ClosedLoop}
\end{figure}

\section{Symmetries in Networked Systems}\label{sec.background}
In this section we provide background on the notion  of symmetries for multi-agent systems originally proposed in~\cite{Sharf2019b}. %This section summarizes the main results of \cite{Burger2014}. 

\subsection{Diffusively Coupled Networks}
In this subsection, we describe the structure of the network dynamical system studied in \cite{Burger2014}. %We also present the connection between networked dynamical systems and network optimization theory.
Consider a collection of agents interacting over a network $\mathcal{G}=(\mathbb{V},\mathbb{E})$.  The nodes $i\in \V$ are assigned dynamical systems $\Sigma_i$, and the edges $e\in\mathbb{E}$ are assigned controllers $\Pi_e$, having the following form:
\begin{align} \label{Dynamics}
\Sigma_i: 
\begin{cases}
\dot{x}_i = f_i(x_i,u_i) \\
y_i = h_i(x_i,u_i),
\end{cases}
\Pi_e: 
\begin{cases}
\dot{\eta}_e= \phi_e(\eta_e,\zeta_e) \\
\mu_e = \psi_e(\eta_e,\zeta_e)
\end{cases}.
\end{align}

We consider stacked vectors of the form $u=[u_1^T,...,u_{|\mathbb{V}|}^T]^T$ and similarly for $y,\zeta$ and $\mu$.  
The network system is diffusively coupled with the controller input described by $\zeta = \mathcal{E}^Ty$, and the control input to each system by $u = -\mathcal{E}\mu$, where $\E$ is an incidence matrix of the graph $\G$.  This structure is denoted by the triplet $(\G,\Sigma,\Pi)$, and is illustrated in Fig. \ref{ClosedLoop}. In this work we focus on homogeneous networks, i.e., where all the agent dynamics and control dynamics are the same.
In \cite{Burger2014} it was shown that 
the network converges to a steady-state satisfying the interconnection constraints if the agents and controllers are (output-strictly) maximum equilibrium independent passive (MEIP).  This property can be thought of as an extension of equilibrium independent passivity developed in \cite{Hines2011}.  The details of these definitions are not essential for the development of this work, and the interested reader is referred to \cite{Burger2014, Sharf2017} for more details. 
%under certain passivity assumptions for the agents and controllers, .  
%
%
For the rest of this paper, we will assume one of the following two alternatives.  If this is not the case, see \cite{Jain2018, Sharf2019a} for plant augmentation techniques.

\begin{assump}\label{assump.1}
The agents $\Sigma_i$ are output-striclty MEIP and the controllers $\Pi_e$ are MEIP.
\end{assump}

\begin{assump}\label{assump.2}
The agents $\Sigma_i$ are MEIP and the controllers $\Pi_e$ are output-strictly MEIP.
\end{assump}

A final technical definition is needed to characterize the  steady-states of the network.  Indeed, embedded in our passivity assumption is a requirement that each agent and controller converge to steady-state outputs given a constant input. This allows us to define  a relation between constant inputs to constant outputs that we call the \emph{steady-state input/output relation} of a system; see \cite{Burger2014}. We denote the steady-state input-output relations of the node $i$ and the edge $e$ by $k_i$ and $\gamma_e$, respectively.  For example, for agent $i$, we say that $(\mathrm{u}_i,\mathrm{y}_i)$ is a steady-state input/output pair if $\mathrm{y}_i \in k_i(\mathrm{u})$.

\begin{defn}[\hspace{0.01pt}\cite{Sharf2019b}]
Two dynamical systems $\Sigma_1$ and $\Sigma_j$ are called \emph{weakly equivalent} if their  steady-state input-output relations are identical.
\end{defn}
Examples of weakly equivalent systems are given in \cite{Sharf2019b}.

\subsection{Group Theory, Graph Automorphisms, and Symmetric MAS}
The main tool we use to facilitate clustering is symmetry, which is usually modelled using the mathematical theory of group theory \cite{Dummit1999}. There are many possible ways to define what a group is, and we choose the most concrete one:
\begin{defn}
Let $X$ be a set, and let $\GG$ be a collection of invertible functions $X\to X$. Then $\GG$ is called a group if for any $\GG \ni f,g: X\to X$, both the composite function $f\circ g$ and the inverse function $f^{-1}$ belong to $\GG$.
\end{defn}
Colloquially, the group $\GG$ defines a collection of symmetries of the set $X$. 
The action of the group $\GG$ on $X$ allows us to identify certain elements of $X$ which are symmetric. 
\begin{defn}
Let $\GG$ be a group of functions $X\to X$. We say that $x,y\in X$ are \emph{exchangeable} (under the action of $\GG$) if there is some $f\in \GG$ such that $f(x)=y$. The \emph{orbit} of $x\in X$ is the set of elements which are exchangeable with $X$.
\end{defn}
Exchangeability was considered in \cite{Sharf2019b} to identify the clustering behavior of a multi-agent system. Namely, the different clusters corresponded to the different orbits of a certain group.
\begin{prop}
Let $\GG$ be a group of functions $X\to X$. Then the set $X$ can be written as the union of disjoint orbits. 
\end{prop}
In this paper, we focus on the automorphism group $\Aut(\G)$ of an (oriented) graph $\G$. Namely, fixing an oriented graph $\G = (\V,\EE)$, a graph automoprhism is a permutation $\psi: \V \to \V$ such that if $(i,j) \in \EE$, then $(\psi(i),\psi(j))\in \EE$. In that case, we abuse the notation and say that $\Aut(\G)$ acts on $\G$ (rather than on $\V$).

Finally, we combine  the notions of automorphisms for graphs with diffusively coupled networks comprised of agents that are weakly equivalent.  

\begin{defn}\label{weakMASaut}
Let $(\G,\Sigma,\Pi)$ be any multi-agent system for SISO agents. A \emph{weak automorphism  of a MAS} is a map $\psi:\mathbb{V}\to\mathbb{V}$ such that the following conditions hold:
\begin{itemize}
\item[i)] The map $\psi$ is an automorphism of the graph $\G$.
\item[ii)] For any $i\in\mathbb{V}$, $\Sigma_i$ and $\Sigma_{\psi(i)}$ are weakly equivalent.
\item[iii)] For any $e \in\mathbb{E}$, $\Pi_e$ and $\Pi_{\psi(e)}$ are weakly equivalent.
\item[iv)] %Moreover, if Assumption \ref{assump.3} does not hold, we demand 
The map $\psi$ preserves edge orientation. 
\end{itemize}
We denote the collection of all weak automorphisms of $(\G,\Sigma,\Pi)$ by $\Aut(\G,\Sigma,\Pi)$. Naturally, this is a subgroup of the group of automorphisms $\Aut(\G)$ of the graph $\G$.
\end{defn}

\section{Cluster Assignment in MAS}\label{sec.cluster}

We now turn our attention to the problem of clustering in MAS. Specifically, we focus on the case where the agents are homogeneous, i.e. they have the exact same model. In that case, we restrict ourselves by requiring that the edge controllers \eqref{Dynamics} are also homogeneous.

In \cite{Sharf2019b} we proposed a symmetry-based framework to understand the clustering behaviour of a multi-agent system, through Definition \ref{weakMASaut}. The main result from \cite{Sharf2019b} can be summarized below.

\begin{thm}[\hspace{0.01pt}\cite{Sharf2019b}] \label{thm.Symmetry}
Consider the diffusively-coupled system $(\G,\Sigma,\Pi)$, and suppose that either Assumption \ref{assump.1} or Assumption \ref{assump.2} hold. Then for any steady-state $\mathrm y=\begin{bmatrix} \mathrm y_1 &  \cdots & \mathrm y_{|\V|}\end{bmatrix}^T$ of the closed-loop and any weak automorphism $\psi\in\Aut(\G,\Sigma,\Pi)$, it follows that $P_\psi \mathrm y = \mathrm y$, where $P_\psi$ is the permutation matrix representation  of $\psi$.
\end{thm}

This result can in fact be used to show that the system converges to a clustering configuration.  The clusters are determined by the orbits of the weak automorphism group. Here, one considers diffusively-coupled networks $(\G,\Sigma,\Pi)$ satisfying either Assumption \ref{assump.1} or \ref{assump.2}. Under this assumption, the closed-loop network is known to converge, and the invariance properties of the steady-state limit are studied. Focusing on homogeneous networks, \cite{Sharf2019b} further identified the value of $\gamma_e(0)$ as decisive. Namely, \cite{Sharf2019b} shows that if $0\in \gamma_e(0)$ for all $e\in \EE$, then the closed-loop network $(\G,\Sigma,\Pi)$ converges to consensus, and otherwise it displays a clustering behavior. Namely, for homogeneous networks, two nodes $i,j \in \V$ are in the same cluster whenever $i,j\in \V$ are exchangable under the action of $\Aut(\G)$, and the converse is almost surely true.

Although \cite{Sharf2019b} presents an extensive study toward understanding the clustering behaviour of diffusively-coupled multi-agent systems, it does not offer a synthesis procedure toward solving the clustering problem:
\begin{prob} \label{prob.cluster}
Consider a collection of $n$ homogeneous agents $\{\Sigma_i\}_{i\in \V}$, and let $r_1,\ldots,r_k$ be positive integers which sum to $n$. Find a graph $\G = (\V,\EE)$ and homogeneous edge controllers $\{\Pi_e\}_{e\in \EE}$ such that the closed loop network converges to a clustered steady-state, with a total of $k$ clusters of sizes $r_1,\ldots,r_k$.
\end{prob}

The goal of this section is use the tools of \cite{Sharf2019b} to solve Problem \ref{prob.cluster}. As described above, this can be achieved in two steps. We first make the following assumption about the controllers:
\begin{assump}\label{assump.3}
The homogeneous MEIP controllers are chosen so that $0\not\in\gamma_e(0)$ for any edge $e\in \EE$.
\end{assump}
Second, given the desired cluster sizes $r_1,\ldots,r_k$, we seek an oriented graph $\G = (\V,\EE)$ such that the action of $\Aut(\G)$ on $\G$ has orbits of sizes $r_1,\ldots,r_k$. Moreover, we desire that the oriented graph $\G$ will be weakly connected\footnote{Recall that a directed graph is weakly connected if its unoriented counterpart is connected.} to assure a flow of information throughout the corresponding diffusively-coupled network. If we find such a graph, the results of \cite{Sharf2019b} guarantee that the desired clustering behavior is achieved almost surely, so long that Assumption \ref{assump.3} holds. We make the following definition for the sake of brevity, and define the corresponding problem:
\begin{defn}
The oriented graph $\G$ is said to be of type OS$(r_1,\ldots,r_k)$ if it is weakly connected and the action of $\Aut(\G)$ on $\G$ has orbits of sizes $r_1,\ldots,r_k$.\footnote{OS stands for "orbit structure"}
\end{defn}
\begin{prob}
Given positive integers $r_1,\ldots,r_k$, determine if an oriented graph of type OS$(r_1,\ldots,r_k)$ exists, and if so, construct it.
\end{prob}

Additionally, running the system with an underlying graph $\G$ requires means to implement the corresponding interconnections. For that reason, graphs with fewer edges are more desirable. We wish to understand how many edges does a directed graph of type OS$(r_1,\ldots,r_k)$ require.  This is formalized in the next result.

\begin{thm} \label{thm.GraphClusterSynthesis}
Let $r_1,\ldots r_k$ be any positive integers, and let $n = r_1+\ldots+r_k$
\begin{itemize}
\item[i)] Any directed graph of type OS$(r_1,\ldots,r_k)$ has at least $m$ edges, where 
\begin{align}
m = \min_{\mathcal{T}\text{\emph{ tree on $k$ vertices}}} \sum_{\{i,j\}\in\mathcal{T}} \lcm(r_i,r_j).
\end{align}
\item[ii)] There exist a directed graph of type OS$(r_1,\ldots,r_k)$ with at most $M$ edges, where
\begin{align} \label{eq.M}
M = \min_{\substack{\mathcal{T}\text{\emph{ path on} }\\ \text{\emph{$k$ vertices}}}} \bigg(\sum_{\{i,j\}\in\mathcal{T}} \lcm(r_i,r_j)\bigg) + \min_{i\in \V} r_i .
\end{align}
\end{itemize} 
\end{thm}

\begin{proof}
We start with the former claim. Consider a graph $\G$ of type OS($r_1,\ldots,r_k$). Let $V_1,\cdots,V_k$ be the orbits of $\Aut(\G)$ in $\G$, corresponding to the different clusters. For any two indices $i,j\in\{1,\cdots,k\}$, we consider the induced bi-partite subgraph $\G_{ij}$ on the vertices $V_i\cup V_j$.\footnote{In other words, only edges between $V_i$ and $V_j$ exist in the subgraph.} We claim that if it is not empty, then $\G_{ij}$ has at least $\lcm(r_i,r_j)$ edges.

Indeed, $V_j$ is invariant to $\Aut(\G)$, hence, for any $x\in V_i$ and any $\psi \in \Aut(\G)$, $x$ and $\psi(x)$ have the same number of edges from them to $V_j$. Thus, all vertices in $V_i$ have the same $\G_{ij}$-degree. Similarly, all vertices in $V_j$ have the same $\G_{ij}$-degree. Let $d_i$ be the $\G_{ij}$-degree of vertices in $V_i$, and $d_j$ be the $\G_{ij}$-degree of vertices in $V_j$. As the edges in $\G_{ij}$ only link $V_i$ and $V_j$, the number of edges in $\G_{ij}$ is equal to $r_i d_i = r_j d_j$. In particular, the number $r_j$ divides $r_i d_i$, and the number $\frac{r_j}{\gcd(r_i,r_j)}$ divides $\frac{r_i}{\gcd(r_i,r_j)} d_i$. However, the numbers $\frac{r_j}{\gcd(r_i,r_j)},\frac{r_i}{\gcd(r_i,r_j)}$ are relatively prime, meaning that $\frac{r_j}{\gcd(r_i,r_j)}$ must divide $d_i$. In particular, $d_i \ge \frac{r_j}{\gcd(r_i,r_j)}$, and $\G_{ij}$ has at least $r_i d_i \ge \frac{r_i r_j}{\gcd(r_i,r_j)} = \lcm(r_i,r_j)$ edges.

Now, consider the condensed graph $\G^{\prime}$. The vertices of $\G^{\prime}$ are $\{v_1,\cdots,v_k\}$, and $v_i$ is connected to $v_j$ if and only if there is an edge between $V_i$ and $V_j$. As $\G$ is weakly connected, $\G^\prime$ is also connected. Let $\mathcal{T}$ be a spanning tree for $\G^\prime$. If there is an edge $e=(i,j)$ in $\mathcal{T}$, $V_i$ and $V_j$ are linked, meaning that the graph $\G_{ij}$ contains at least $\lcm(r_i,r_j)$ edges. As a result, $\G$ has at least $ \sum_{\{i,j\}\in\mathcal{T}} \lcm(r_i,r_j) \ge m$ edges.

We now consider the second part of the theorem. Let $\mathcal{T}$ be a path on $k$ vertices, and let $i_1,\cdots,i_k$ be the order of the vertices in the path. By reordering $r_1,\cdots,r_k$, we may assume that $i_j = j$ for $j=1,2,\cdots,k$. Let $i$ be the vertex at which $r_i$ is minimized. For each $j\in \{1,\cdots,k\}$, we number the vertices in $V_j$ as ${v_1^j,\cdots,v_{r_j}^j}$. Build the graph $\G$ as follows:
\begin{itemize}
\item[I)] For every $j$, add the following edges between $V_j$ and $V_{j+1}$: For $p=1,\ldots, \lcm(r_j,r_{j+1})$, connect $v^j_{p\mod r_j}$ to $v^{j+1}_{p\mod r_{j+1}}$.
\item[II)] If $r_i > 1$, then for each $p=1,\ldots,r_i$, connect $v^i_p$ to $v^i_{(p+1)\mod r_i}$.
\end{itemize}
The number of edges in $\G$ is equal to $\sum_{\{i,j\}\in\mathcal{T}} \lcm(r_i,r_j) + r_i$ if $r_i \ge 2$, and equal to $\sum_{\{i,j\}\in\mathcal{T}} \lcm(r_i,r_j)$ otherwise. We show that the orbits of the action of $\Aut(\G)$ on $\G$ are given by $V_1,\cdots, V_k$, and that $G$ is weakly connected, concluding the proof.

First, we show $\G$ is weakly connected. We note that the induced subgraph on $V_i$ is weakly connected. Indeed, this is clear if $r_i =1$, and in the case $r_i \ge 2$, the following cycle eventually passes through all the nodes in $V_i$:
$$
v^i_1 \to v^i_2 \to v^i_3 \to \cdots
$$
Now, by the construction of $\G$, any vertex $v^i_p$ is connected all vertices $v^{l}_{p \mod r_l}$. Thus, for any two arbitrary vertices $v^{j_1}_{p_1}$ and $v^{j_2}_{p_2}$, we can consider a path that starts from $v^{j_1}_{p_1}$, goes to $v^{i}_{p_1\mod r_i}$, moves to $v^{i}_{p_2\mod r_{i}}$ using the connectivity of induced subgraph on $V_i$, and continues to $v^{j_2}_{p_2}$. Thus $\G$ is weakly connected.

As for the orbits, we first consider the map $\psi$ defined on the vertices of $\G$ by sending each vertex $v^j_p$ to $v^j_{(p+1)\mod r_j}$. By definition of the graph $\G$, it is clear that $\psi$ is a graph automorphism. Moreover, repeatedly applying $\psi$ can move $v^j_p$ to any vertex in $V_j$. Thus the orbit of $v^j_p$ contains the set $V_j$.

Conversely, we show that each $V_j$ is invariant under $\Aut(\G)$, so the orbits of the action of $\Aut(\G)$ on $\G$ are exactly $V_1,\cdots V_k$. This is obvious if $k=1$, as then $V_1 = \V$, so we assume that $k\ge 2$. In that case, either $i\neq 1$ or $i\neq k$ (or both). We assume $i\neq k$, as the complementary case can be treated similarly. Graph automorphisms preserve all graph properties, and in particular, they preserve the out-degree of vertices. As all edges are oriented from $V_j$ to $V_{j+1}$ or from $V_i$ to itself where $i\neq k$, the vertices in $V_k$ have an out-degree of $0$, and they are the only ones with this property. Thus $V_k$ must be invariant under $\Aut(\G)$. The only vertices with edges to $V_k$ are in $V_{k-1}$, meaning that $V_{k-1}$ is also invariant under $\Aut(\G)$. Iterating this argument, we conclude that $V_1,\cdots V_k$ must all be invariant under the action of $\Aut(\G)$. Thus they are the orbits of the action, and the action of $\Aut(\G)$ on $\G$ has orbits of size $r_1,\cdots,r_k$. This completes the proof.
\end{proof}
The construction of OS-type graphs presented in the proof is explicitly stated in Algorithm \ref{alg.BuildingGraphs}. The Algorithm takes as input the desired cluster sizes, as well as a path $\mathcal{T}$ describing the desired interconnection topology of the different clusters in the graph. While any path gives a graph of type $OS(r_1,\ldots,r_k)$, the number of edges in the graph depends on the path $\mathcal{T}$.
\begin{algorithm} [h]
\caption{Building $OS$-type graphs}
\label{alg.BuildingGraphs}
{\bf Input:} A collection $r_1,\ldots,r_k$ of positive integers summing to $n$, and a path $\mathcal{T}$ on $k$ vertices.\\
{\bf Output:} A graph $\G$ of type $OS(r_1,\ldots,r_k)$.
\begin{algorithmic}[1]
\State Let $\G = (\V,\E)$ be an empty graph.
\For{$j=1,\ldots,k$ and $p=1,\ldots,r_j$}
\State Add a node with label $v^j_p$ to $\V$.
\EndFor
\For{any edge $\{i,j\}$ in $\mathcal{T}$} 
\For{$p=1,\ldots,\lcm(r_i,r_j)$}
\State Add the edge $v^i_{p \mod r_i} \to v^j_{p \mod r_j}$ to $\EE$.
\EndFor
\EndFor
\State Compute $i^\star = \arg\min\{r_i\}$. If $r_{i^\star} = 1$, go to step 10.
\For{$p=1,\ldots,r_{i^\star}$}
\State Add the edge $v_p^{i^\star} \to v_{(p+1) \mod r_{i^\star}}^{i^\star}$ to $\EE$.
\EndFor
\State {\bf Return} $\G = (\V,\EE)$.
\end{algorithmic}
\end{algorithm}

\begin{rem}
One might ask why the lower bound considers all possible trees, while the upper bound only considers path graphs. The main reason for this distinction can be seen in the proof - one can build a general graph $\G$ using a tree graph as a basis, instead of the path graph $1\to 2 \to 3 \to\cdots\to k$. In that case, proving that the orbit of $v^j_p$ contains $V_j$ is easy, but it is possible that the orbit is actually larger than $V_j$.
\end{rem}

\begin{rem}
The lower bound in Theorem \ref{thm.GraphClusterSynthesis} can be found using any algorithm finding a minimal spanning tree, e.g. Kruskal's algorithm which runs in polynomial time \cite{Cormen2009}. However, the upper bound requires one to solve a variant of the traveling salesman problem, which is known to be {\bf NP}-hard \cite{Cormen2009}. 
\end{rem}

Theorem \ref{thm.GraphClusterSynthesis} deals with a general cluster assignment by explicitly constructing a graph solving the problem. We apply it to more specific cases in order to achieve concrete bounds on the number of edges needed for clustering in these cases.
%\todo[inline]{can the figures be more illustrative? can you color code nodes that will form in different clusters? can you also run simulations to show how they cluster (maybe color coding trajectories to match nodes in graph) \textcolor{blue}{Done (appears later in the section, reference after the theorem)}}

\begin{cor}
Suppose that all cluster sizes $r_1,\cdots,r_k$ are equal, and bigger than $1$. Then there exists a graph of type OS($r_1,\ldots,r_k$) with at most $n = r_1+\ldots+r_k$ edges.
\end{cor}

\begin{proof}
Let $r$ be the size of all clusters. We note that in this case, $\lcm(r,r) = r$, and the number of clusters is $k=n/r$. Thus, if we build a graph $\G$ as in the proof of Theorem \ref{thm.GraphClusterSynthesis}, then $\G$ has exactly $n$ edges, as the summation over the edges has $k-1$ elements. It remains to show that no such graph with less than $n$ edges exists.

First, any graph on $n$ nodes and less than $n-1$ edges is not weakly connected \cite{Bondy1976}. Thus, it suffices to show that no such graph with $n-1$ edges exists. Indeed, let $d_i$ be the out-degree of any vertex in the $i$-th orbit, $i=1,\cdots,k$. As the action of $\Aut(\G)$ preserves the out-degree, it is the same for all vertices in the $i$-th orbit. The number of edges is the sum of the out-degree over all vertices, but it's equal to $rd_1 + \cdots + rd_k$. Thus the number of edges in $\G$ must be divisible by $r$. But $n = kr$, meaning that $n-1$ is not divisible by $r$ unless $r=1$. Thus there is no such graph on $n-1$ edges.
\end{proof}
%\todo[inline]{can you construct an example here illustrating this? It is somehow confusing to think about a weakly connected digraph with $n$ edges (must be a cycle) having multiple clusters of equal size.  The corollary assumes there are $k$ clusters, right? but in proof it seems you arrive at $k$ by looking at the size the cluster should be \textcolor{blue}{If we know the number of vertices, which is $n$, then $k$ and $n/k$ are interchangeable I added an example of two such graphs that are achieved for $n=6$}}
\begin{cor}
Let $r_1,\cdots,r_k$ be positive integers such that $k\ge 2$ and that for every $j,l$, either $r_j$ divides $r_l$ or vice versa. Then there exists a graph of type OS($r_1,\ldots,r_k$) with $n = r_1+\ldots+r_k$ edges.
\end{cor}

\begin{proof}
We reorder the numbers $r_1,\cdots, r_k$ such that $r_l$ divides $r_j$ for $l\le j$. We note that if $r_l$ divides $r_j$, then $\lcm(r_l,r_j) = r_l$. Thus, if we let $\G$ be the graph built in the proof of Theorem \ref{thm.GraphClusterSynthesis}, it has the following number of edges:
$$
\sum_{j = 1}^{k-1} \lcm(r_j,r_{j+1}) + r_1 = \sum_{j=1}^{k-1} r_{j+1} + r_1=\sum_{j=1} r_j = n.
$$
\end{proof}

\begin{cor}
Let $r_1,\cdots,r_k$ be positive integers such that $r_j \le q$ for all $j$, and let $n = r_1+\cdots+r_k$. Then there exists a graph of type OS($r_1,\ldots,r_k$) with at most $n+O(q^3)$ edges.
\end{cor}

\begin{proof}
By reordering, we may assume that that $r_1\le r_2\le \cdots r_k$. Let $m_l$ be the number of clusters of size $s$ for $l=1,2,\cdots,q$. As before, consider the graph $\G$ built in the proof of Theorem \ref{thm.GraphClusterSynthesis}. If $r_j=r_{j+1}$ then $\lcm(r_j,r_{j+1}) = r_l$, and $\lcm(r_j,r_{j+1}) \le r_j r_{j+1}$ otherwise. Thus, the number of edges in $\G$ is given by:
\begin{align*}
\sum_{j =1}^{k-1} \lcm(r_j,r_{j+1}) +r_1 \le  \hspace{-7pt}\sum_{\substack{l\in\{1,\ldots,q\},\\ m_l \neq 0}} \hspace{-7pt}(m_l-1)l + \sum_{l=1}^{q-1} l(l-1) + r_1.
\end{align*}
Indeed, for each $l\in\{1,\ldots,q\}$, if there's at least one cluster of size $l$, then there are $m_l -1$ edges in the path $\mathcal{T} = 1\to 2 \to \cdots \to k$ that touch two clusters of size $l$. The second term bounds the number of edges that appear between clusters of different sizes. We note that $n = \sum_{l=1}^q l m_l$, so the first term is bounded by $n$. As for the second term, we can bound $l(l-1)$ by $l^2$ and then use the formula $\sum_{l=1}^{q-1} l^2 = \frac{(q-1)q(2q-1)}{6}$, dating back to Fibonacci's \emph{Liber Abaci} from 1201 \cite{Sigler2003}. Lastly, the last term $r_1$ is bounded by $q$. This completes the proof.
\end{proof}

We now give examples of graphs constructed by Algorithm \ref{alg.BuildingGraphs}, and show they indeed force a clustering structure on the agents.
\begin{exam}\label{SYMExam.ClusteredGraphSynthesis}
We consider a collection of identical agents, all of the form $\dot{x} = -x+u+\alpha ,y=x$ where $\alpha$ is a log-uniform random variable between $0.1$ and $1$, identical for all agents. In all experiments described below, we considered identical controllers, equal to the static nonlinearity of the form
$$
\mu = a_1 + a_2 (\zeta + \cos(\zeta)),
$$
where $a_1$ was chosen as a Gaussian random variable with mean $0$ and standard deviation $10$, and $a_2$ was chosen as a log-uniform random variable between $0.1$ and $10$. We note that the agents are indeed output-strictly MEIP and the controllers are MEIP. Moreover, the network is homogeneous, so $\Aut(\G,\Sigma,\Pi) = \Aut(\G)$. Thus, we can use the graphs constructed by Theorem \ref{thm.GraphClusterSynthesis}, as depicted in Algorithm \ref{alg.BuildingGraphs}, to force a clustering behavior.

We first consider a network of $n=15$ agents and tackle the cluster synthesis problem with five equally-sized clusters, i.e., $r_1=r_2=r_3=r_4=r_5=3$. One possible graph forcing these clusters, as constructed by Algorithm  \ref{alg.BuildingGraphs} for a path $\mathcal{T}$ of the form $1\to 2\to 3\to 4\to 5$, can be seen in Figure \ref{fig.ClusterSynthesisTheoremGraph1}, along with the agents' trajectories for the closed-loop system. 

Second, we consider a network of $n=10$ agents with desired cluster sizes $r_1=1,r_2=2,r_3=3,r_4=4$. We build a graph forcing these cluster sizes by using Algorithm  \ref{alg.BuildingGraphs} for a path $\mathcal{T}$ of the form $4\to 2\to 1\to 3$, which is the minimizer in \eqref{eq.M}. The graph can be seen in Figure \ref{fig.ClusterSynthesisTheoremGraph1}, along with the agents' trajectories for the closed-loop system. 
 \begin{figure}[t!]
 \begin{center}
 	\hfill\subfigure[Graph forcing cluster sizes $r_1=r_2=r_3=r_4=r_5=3$. Nodes with the same color will be in the same cluster.] {\scalebox{.80}{\includegraphics{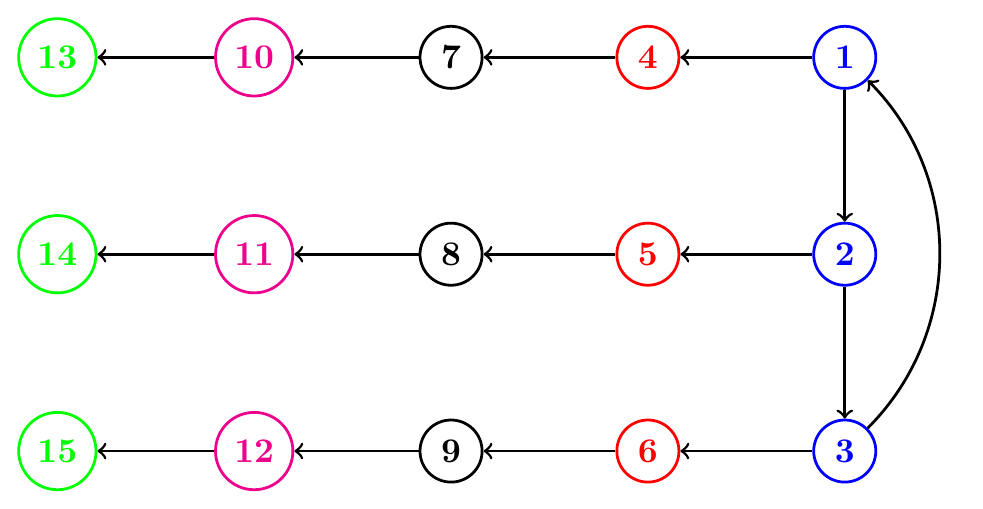}}}\hfill \vspace{0.2cm}
	\subfigure[Agent's trajectories for the closed-loop system. Colors correspond to node colors in the graph.] {\scalebox{.5}{\includegraphics{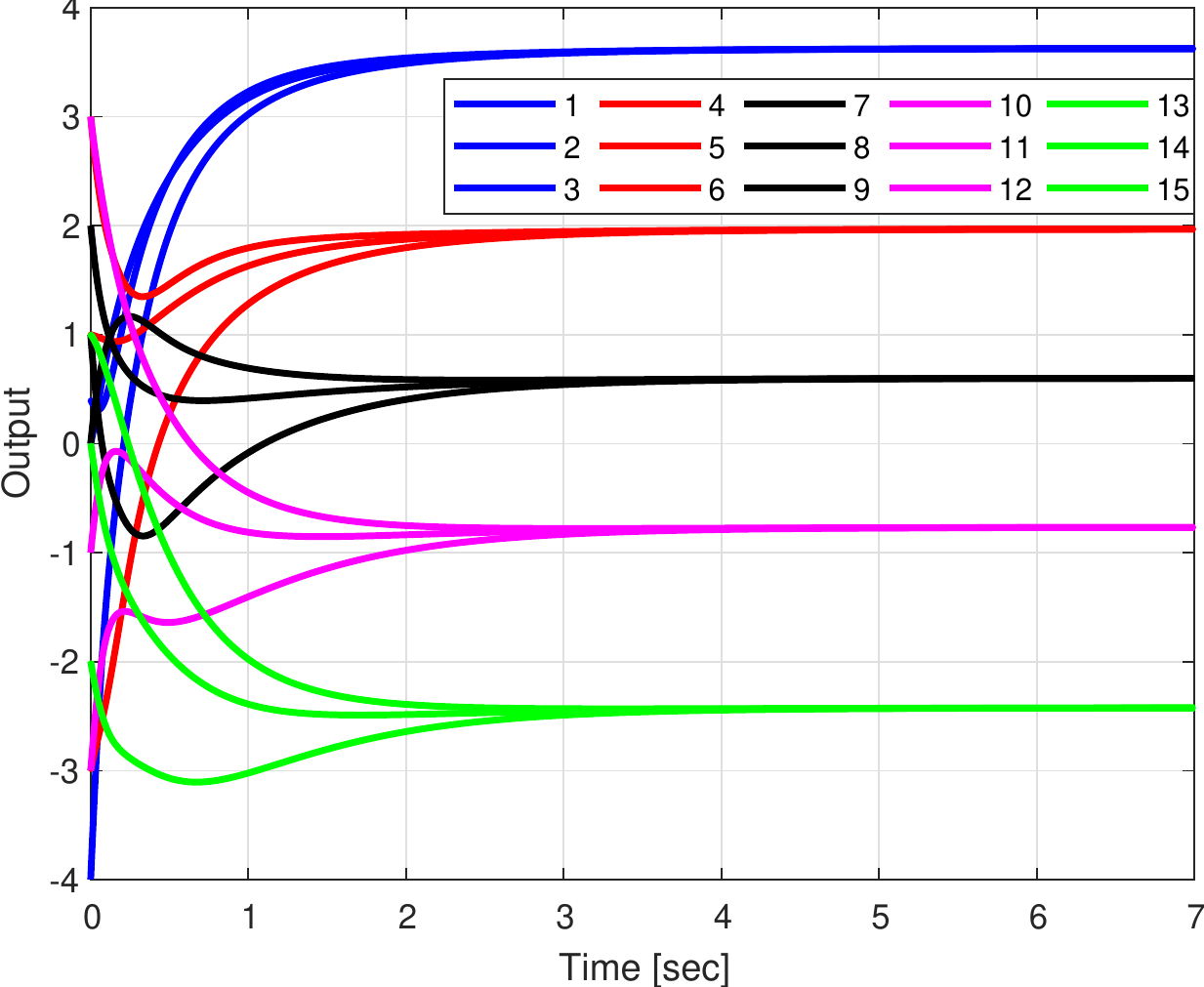}}}
   \caption{First example of graphs solving the cluster synthesis problem, achieved by running Algorithm \ref{alg.BuildingGraphs}.}
	\label{fig.ClusterSynthesisTheoremGraph1}
	\vspace{-15pt}
 \end{center}
 \end{figure}
 \begin{figure}[t!]
 \begin{center}
 	\hfill\subfigure[Graph forcing cluster sizes $r_1=1,r_2=2,r_3=3,r_4=4$. Nodes with the same color will be in the same cluster.] {\scalebox{.80}{\includegraphics{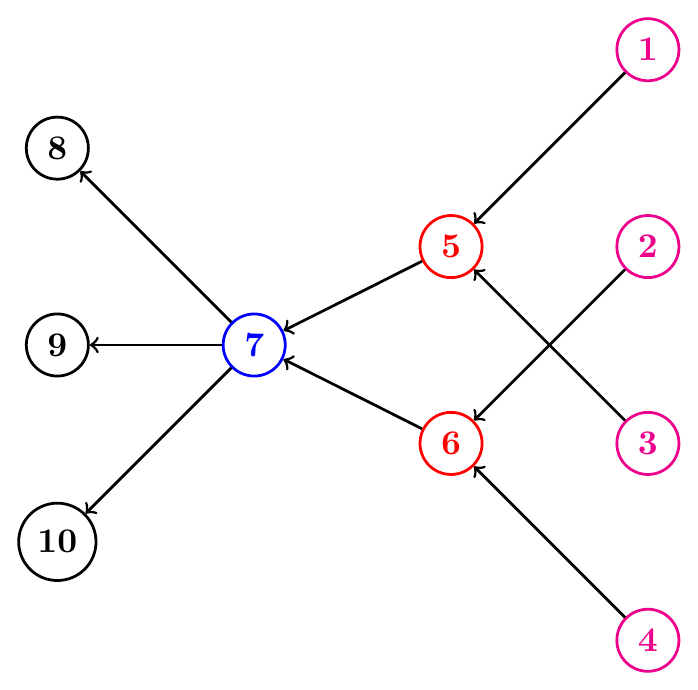}}}\hfill \vspace{0.2cm}
	\subfigure[Agent's trajectories for the closed-loop system. Colors correspond to node colors in the graph.] {\scalebox{.5}{\includegraphics{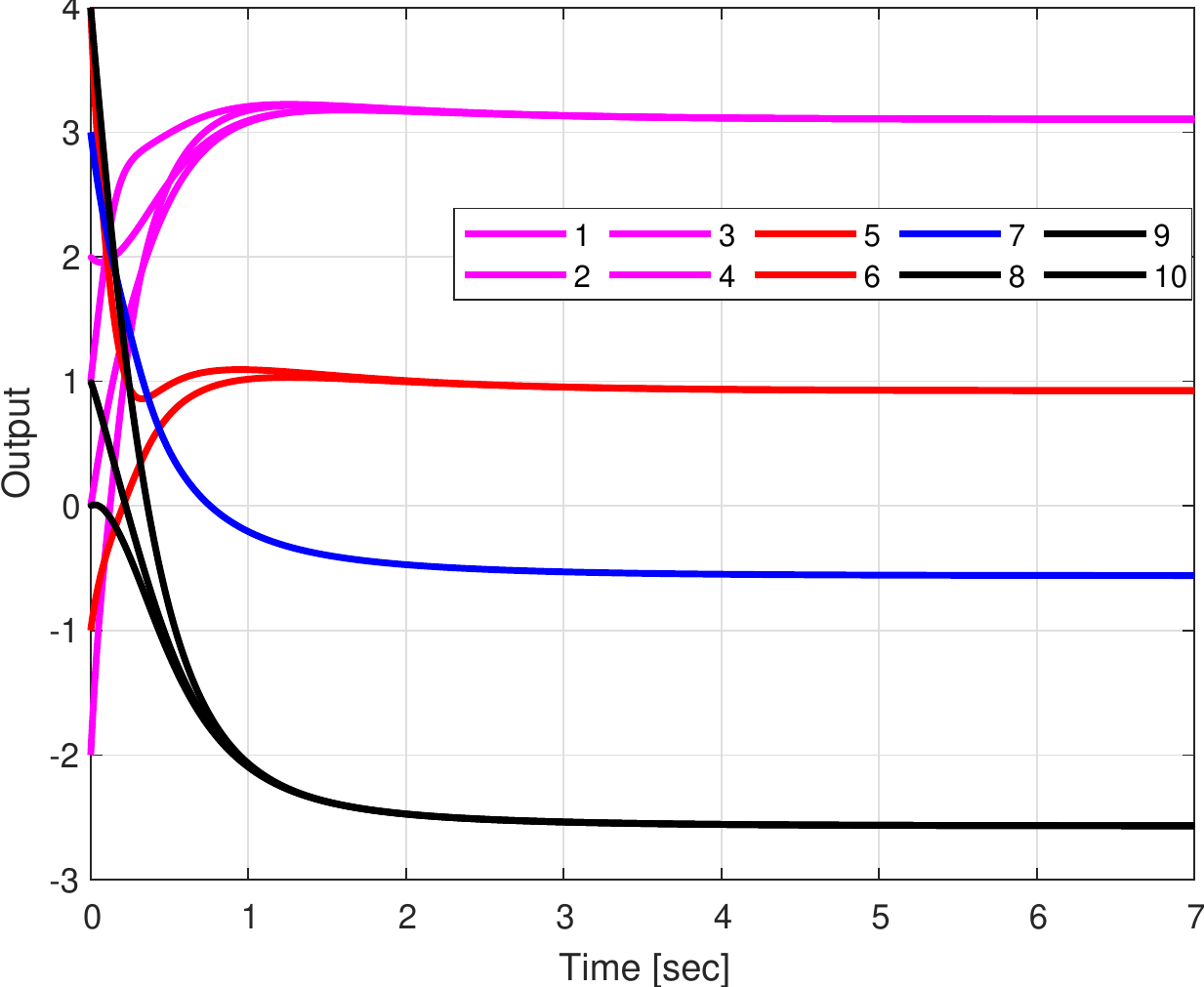}}}
   \caption{First example of graphs solving the cluster synthesis problem, achieved by running Algorithm \ref{alg.BuildingGraphs}.}
	\label{fig.ClusterSynthesisTheoremGraph2}
	\vspace{-15pt}
 \end{center}
 \end{figure}

\end{exam}

\section{Conclusions}\label{sec.conclusion}

In this work we explored the problem of cluster assignment for homogeneous multi-agent systems.  We extended the results of \cite{Sharf2019b} to provide a characterization of graphs with an automorphism group containing orbits of a prescribed size.  When such graphs are used in a network comprised of weakly equivalent agent and controller dynamics, the network converges to a cluster configuration.

%=== Some discussion is needed here on the fact that passivity/SISO isn't needed (and referencing our other papers on the subject). In fact, what is needed here is a way to guarantee global asymptotic convergence of the closed-loop system...

\bibliographystyle{IEEEtran}
\bibliography{main}

\end{document}